\documentclass[conference]{IEEEtran}

\usepackage{latexsym, amssymb,graphicx, amsmath, subfigure,threeparttable}
\usepackage{amsthm}

\newtheorem{Theorem}{Theorem}
\newtheorem{Lemma}{Lemma}

\newtheorem{Proposition}{Proposition}

\newtheorem{Remark}{Remark}

\newtheorem{Example}{Example}

\onecolumn

\begin{document}
%
\title{On the Weight Hierarchy of Locally \\ Repairable Codes}


\author{\IEEEauthorblockN{Jie Hao\IEEEauthorrefmark{1}, Shu-Tao Xia\IEEEauthorrefmark{1}, Bin Chen\IEEEauthorrefmark{1}\IEEEauthorrefmark{2}, and Fang-Wei Fu  \IEEEauthorrefmark{3}}
\IEEEauthorblockA{ \IEEEauthorrefmark{1}Graduate School at Shenzhen, Tsinghua University, Shenzhen, China\\
\IEEEauthorrefmark{2}School of Mathematical Sciences, South China Normal University, Guangzhou, China\\
\IEEEauthorrefmark{3}Chern Institute of Mathematics and LPMC, Nankai University, Tianjin, China \\
Email: j-hao13@mails.tsinghua.edu.cn, xiast@sz.tsinghua.edu.cn, binchen14scnu@m.scnu.edu.cn,  fwfu@nankai.edu.cn}
}

\maketitle

\begin{abstract}
An $(n,k,r)$ \emph{locally repairable code} (LRC) is an $[n,k,d]$ linear code where every code symbol can be repaired from at most $r$ other code symbols. An LRC is said to be optimal if the minimum distance attains the Singleton-like bound  $d \le n-k-\lceil k/r \rceil +2$. The \emph{generalized Hamming weights} (GHWs) of linear codes are fundamental parameters which have many useful applications. Generally it is difficult to determine the GHWs of linear codes. In this paper, we  study the GHWs of LRCs. Firstly, we obtain  a generalized Singleton-like bound on the $i$-th $(1\le i \le k)$  GHWs of general $(n,k,r)$ LRCs. Then, it is shown that for an optimal $(n,k,r)$ LRC with  $r \mid k$, its weight hierarchy  can be completely determined, and the $i$-th GHW of an optimal $(n,k,r)$  LRC with  $r \mid k$ attains the proposed generalized Singleton-like bound for all $1\le i \le k$. For an optimal $(n,k,r)$ LRC with $r \nmid k$, we give  lower bounds on the GHWs of the  LRC and its dual code. Finally, two general bounds on linear codes in terms of the GHWs are presented. Moreover, it is also shown that some previous results on the bounds of minimum distances of linear codes can also be explained or refined in terms of the GHWs.

\end{abstract}

\section{Introduction}
Locally repairable codes have attracted a lot of interest recently. An $(n,k,r)$ LRC is an $[n,k,d]$ linear code with locality constraint on  the code symbols, i.e., each of the $n$ code symbols can be repaired by accessing at most $r$ other code symbols. The minimum distance of an LRC satisfies the well-known Singleton-like bound \cite{Gopalan}
\begin{equation} \label{singleton-like-bound}
 d \leq n -k - \left\lceil \frac{k}{r} \right\rceil + 2.
\end{equation}
When $r=k$, the above bound reduces to the classical Singleton bound $d\le n-k+1$ \cite{MacWilliams}.
Further studies on the bounds of LRCs can be found  in \cite{Hao}-\cite{Tamo-RS}.
Many works have proposed optimal LRCs attaining the Singleton-like bound (\ref{singleton-like-bound}), e.g., \cite{Tamo-RS}-\cite{binaryLRC}, among which codes with small field size are of particular interest.
The elegant Reed-Solomon-like optimal LRCs proposed by Tamo and Barg in \cite{Tamo-RS}  require the field size to be just slightly greater than the code length. The codes were further generalized to optimal cyclic LRCs in \cite{Tamo-cyclic} and LRCs on algebraic curves \cite{Barg-curve}.
All the possible four classes of optimal binary LRCs attaining the  bound (\ref{singleton-like-bound}) were found in \cite{binaryLRC}.
A bound of linear codes with a local-error-correction property were proposed in \cite{delta-bound} which could include  the bound (\ref{singleton-like-bound}) as a special case.

The \emph{generalized Hamming weight} (GHWs) \cite{Wei}\cite{Helleseth} are fundamental parameters of linear codes and were first used by Wei  in cryptography to fully characterize the performance of linear codes  in a wire-tap channel of type II \cite{Wei}.
Let $\mathcal{C}$ be a $q$-ary $[n,k,d]$ linear code and  $\mathcal{D}$ be a subcode of $\mathcal{C}$. Let ${\rm dim}(\mathcal{D})$ denote the dimension of $\mathcal{D}$.
 The \emph{support} of a vector is the set of coordinates of its non-zero components. If a coordinate is in the support of a vector, it is said to \emph{be covered} by the vector.
The  \emph{support} of $\mathcal{D}$ is defined to be
 \begin{equation*}
 {\rm supp}(\mathcal{D}) = \{i: \exists (c_1, c_2, \cdots, c_n)\in \mathcal{D}, c_i \neq 0\}.
 \end{equation*}
For $1 \le i \le k$, the $i$-th  GHW  of $\mathcal{C}$ is defined to be
 \begin{equation*}
 d_i= {\rm min} \{ | {\rm supp}(\mathcal{D}) | : \; \mathcal{D} \subset \mathcal{C} \; {\rm and} \; {\rm dim}(\mathcal{D}) = i \},
 \end{equation*}
 where $|{\rm supp}(\mathcal{D})|$ denotes the cardinality of the set ${\rm supp}(\mathcal{D})$.
Note that $ d_1$ is the minimum distance of $\mathcal{C}$. The $i$-th $(1 \le i \le k)$ GHW of $\mathcal{C}$ satisfies the generalized Singleton bound
\begin{equation}\label{generalized-singleton-bound}
  d_i \le n-k+i.
\end{equation}
MDS codes meet the generalized Singleton bound for all $1 \le i \le k$.
The \emph{weight hierarchy} of $\mathcal{C}$ is  the set of integers $\{d_i \; | \; 1 \le i \le k\}$. For an $[n,k]$ linear code $\mathcal{C}$, $1 \le d_1 < d_2 < \cdots < d_k =n$. The \emph{gap numbers} of $\mathcal{C}$ is  the complement of the set of its weight hierarchy, denoted as $\{g_i, 1 \le i \le n-k\} = [n] \setminus \{d_i \; | \; 1 \le i \le k\}$. Let $\mathcal{C}^{\bot}$ be the dual code of  $\mathcal{C}$ and its weight hierarchy and gap numbers respectively be $\{d_i^{\bot} \; | \; 1 \le i \le n-k\}$ and $\{g_i^{\bot} \; | \; 1 \le i \le k\}$. The  weight hierarchy of $\mathcal{C}$ and $\mathcal{C}^{\bot}$ satisfies the following property of duality.
\begin{equation}\label{weight-hierarchy-duality}
  \{d_i  \; | \; 1 \le i \le k\}=  \{1,2,\cdots,n \} \setminus \{n+1-d_j^{\bot} \; | \; 1 \le j \le n-k\}.
\end{equation}
In terms of the gap numbers of $\mathcal{C}^{\bot}$,
\begin{equation}\label{gap-number-dual}
    d_i  =  (n+1) - g_{k-i+1}^{\bot}, \; \;  1 \le i \le k.
\end{equation}
The minimum distance  $d = d_1  =  (n+1) - g_{k}^{\bot},$
where 
\begin{eqnarray}
  g_k^{\bot} &=&{\rm{max}}\{k+i: 1\le i \le n-k, d_i^{\bot} < k+i\}, \label{g_k_1}  \\
   &=& {\rm{min}}\{k+i: 1\le i \le n-k, d_i^{\bot} = k+i\} -1. \label{g_k_2}
\end{eqnarray}

 Many research works have been devoted to determine or estimate the GHWs of many series of linear codes, such as Hamming codes \cite{Wei}, Reed-Muller codes \cite{Wei,Heijnen}, BCH codes and their dual codes \cite{Van}-\cite{Vlugt}, etc..
 Generally speaking, it  is difficult to determine the GHWs of linear codes, and the complete weight hierarchy is known for only a few cases.
Recently some works  started to analyze the GHWs of particular LRCs.  The GHWs of LRCs on algebraic curves were studied in \cite{Ballico}.   Lalitha \emph{et al.} \cite{Lalitha} studied the GHWs and weight distributions of  maximally recoverable codes, which are variants of optimal LRCs \cite{Goppa-IT}.

In this paper, we study the GHWs of  LRCs.
Firstly, we obtain a generalized Singleton-like bound on the $i$-th ($1 \le i \le k$) GHWs of general $(n,k,r)$ LRCs. When $i = 1$, the proposed generalized Singleton-like bound gives the Singleton-like bound (\ref{singleton-like-bound}). When $r=k$, i.e., there is no locality constraint, the proposed generalized Singleton-like bound reduces to the classical generalized Singleton bound (\ref{generalized-singleton-bound}).
Then, for optimal $(n,k,r)$  LRCs attaining the Singleton-like bound (\ref{singleton-like-bound}), some lower  bounds on the  GHWs of optimal LRCs and their dual codes are obtained.
Surprisingly, it is shown that for an $(n,k,r)$ optimal LRC $\mathcal{C}$ meeting the Singleton-like bound (\ref{singleton-like-bound}) with $r \mid k$,
its weight hierarchy  can be  completely determined as
\begin{equation*}
  d_i= n-k-\left \lceil \frac{k-i+1}{r} \right \rceil + i + 1,  \; \; 1 \le i \le k.
\end{equation*}
Its dual code $\mathcal{C}^{\bot}$ has weight hierarchy
\begin{eqnarray*}
  d_i^{\bot}=
  \left\{
  \begin{array}{ll}
  i(r+1), \quad & \mbox{$1 \le i \le k/r$},\\
  k+i, \quad &\mbox{$ k/r < i \le n-k$}.
  \end{array}
  \right.
\end{eqnarray*}
Finally, we employ a parity-check matrix approach \cite{Hao}  to give  general upper bounds of linear codes in terms of the GHWs of their dual codes, which can  include several known bounds of LRCs as special cases. It is also shown that some results proposed in \cite{Tamo-matroid}\cite{Wang:IT} on bounds of minimum distances of linear codes  can  be explained or refined in terms of the GHWs.

The rest of this paper is organized as follows.
In Section \ref{sec:upper-bound}, upper bounds on the GHWs of general LRCs are presented.
Section~\ref{sec:lower-bound} studies the GHWs of optimal LRCs.
Section~\ref{sec:general-bound} analyzes the bounds of the linear codes in terms of the GHWs of the dual codes.
Section \ref{sec:conclusion} concludes the paper.

\section{Upper Bounds on the GHWs of  General LRCs} \label{sec:upper-bound}
In this section, we give some upper bounds on the GHWs of LRCs and their dual codes. Specially a generalized Singleton-like bound on the $i$-th ($1 \le i \le k$) GHWs of general $(n,k,r)$ LRCs is obtianed. The proposed generalized Singleton-like bound can give the Singleton-like bound (\ref{singleton-like-bound}) when $i=1$ and  reduce to the classical generalized Singleton bound (\ref{generalized-singleton-bound}) when there is no locality constraint.

\begin{Lemma} \label{thm-ghw-upper-bound-dual}
Let  $\mathcal{C}$ be an $(n,k,r)$ LRC and $\mathcal{C}^{\bot}$ be its dual code. The $i$-th ($1 \le i \le n-k$) GHW of the dual code $\mathcal{C}^{\bot}$ satisfies
\begin{eqnarray}\label{ghw-upper-bound-dual}
  d_i^{\bot} \le
  \left\{
  \begin{array}{ll}
  i(r+1), \quad & \mbox{$1 \le i \le \lfloor k/r \rfloor$},\\
  k+i, \quad &\mbox{$ \lfloor  k/r  \rfloor + 1 \le i \le n-k$}.
  \end{array}
  \right.
\end{eqnarray}
\end{Lemma}

\begin{proof}
Since every code symbol in an $(n,k,r)$ LRC $\mathcal{C}$ has locality $r$, each coordinate is covered by at least one  parity-check equation with weight at most $r+1$ from  $\mathcal{C}^{\bot}$. By \cite{Hao}, $\lfloor  k/r  \rfloor \le n-k$. For  $1 \le i \le \lfloor  k/r  \rfloor$, let us select a subcode $\mathcal{D}^{\bot}$ of $\mathcal{C}^{\bot}$ with dimension $i$.  For the first coordinate, select a parity-check equation with weight at most $r+1$ to cover it; then, for the first uncovered coordinate, select another parity-check equation with weight at most $r+1$ to cover it; repeating the procedure iteratively until we obtain $i$ parity-check equations. Apparently, these $i$ parity-check equations are linearly independent, which implies that these $i$ vectors are a basis of a subcode $\mathcal{D}^{\bot}$  with dimension $i$. Since each of these $i$ vectors has weight $\le r+1$,  $|{\rm supp}(\mathcal{D}^{\bot})| \le i(r+1)$.
The  generalized Singleton bound says that $d_i^{\bot} \le k+i$. 
For $1 \le i \le \lfloor  k/r  \rfloor$, it holds that $i(r+1) \le k+i$. Hence for $1 \le i \le \lfloor  k/r  \rfloor$, $d_i^{\bot} \le |{\rm supp}(\mathcal{D}^{\bot})|  \le i(r+1).$  For $ \lfloor  k/r  \rfloor + 1 \le i \le n-k$,  since  $i(r+1) > k+i$, we have $d_i^{\bot} \le k+i$. Thus the conclusion follows.
\end{proof}

\smallskip
\begin{Lemma} \label{lemma-ghw-lower-bound-dual}
 Let $\mathcal{C}$ be an $(n,k,r)$ LRC and  $\mathcal{C}^{\bot}$ be its dual code. For $1 \le i \le \lfloor  k/r  \rfloor$, the  GHWs of the dual code $\mathcal{C}^{\bot}$ satisfies
\begin{equation}\label{ghw-upper-bound-dual}
  d_{i+1}^{\bot} \le d_{i}^{\bot} + (r+1).
\end{equation}
\end{Lemma}
\begin{proof}
Suppose $\mathcal{D}_i^{\bot}$ is a subcode  of $\mathcal{C}^{\bot}$ with dimension $i$ and  $|{\rm supp}(\mathcal{D}_i^{\bot})| = d_i^{\bot}$. Let  $\mathbf{H}_i=[(\mathbf{h}_1^T, \cdots \mathbf{h}_i^T)^T]$ be a basis of  $\mathcal{D}_i^{\bot}$ with $i$ independent parity-check equations from $\mathcal{D}_i^{\bot}$. Choose a coordinate $j$ from the set $[n] \setminus {\rm supp}(\mathcal{D}_i^{\bot})$. Then select a parity-check equation $\bar{\mathbf{h}}$ with weight at most $r+1$ to cover the coordinate $j$.  Apparently $\bar{\mathbf{h}}$ is independent with the vectors in $\mathbf{H}_i$, which indicates that $\mathbf{H}_{i+1} = [\mathbf{H}_i^{T}, \bar{\mathbf{h}}^T]^T$ is the basis of  a subcode $\mathcal{D}_{i+1}^{\bot}$  of $\mathcal{C}^{\bot}$  with dimension $i+1$. Since $|{\rm supp}(\bar{\mathbf{h}})| \le r+1$, we have $|{\rm supp}(\mathcal{D}_{i+1}^{\bot})| \le d_i^{\bot} + (r+1)$. Thus it follows that $ d_{i+1}^{\bot} \le |{\rm supp}(\mathcal{D}_{i+1}^{\bot})| \le  d_{i}^{\bot} + (r+1).$
\end{proof}

\smallskip

\begin{Lemma} \label{coro-ghw}
Let  $\mathcal{C}$ be an $(n,k,r)$ LRC and $\mathcal{C}^{\bot}$ be its dual code. If the $i$-th ($1 < i \le \lfloor  k/r  \rfloor$) GHW of  $\mathcal{C}^{\bot}$ satisfies that $d_i^{\bot} = i(r+1).$ Then for any $1\le j < i$, the $j$-th  GHW $d_j^{\bot} = j(r+1).$
\end{Lemma}
\begin{proof}
By Lemma \ref{lemma-ghw-lower-bound-dual}, $d_{j}^{\bot} \ge d_{j+1}^{\bot} - (r+1)$, $1 \le j \le \lfloor k/r  \rfloor -1$.
Combining that $d_i^{\bot} = i(r+1)$, we obtain that for any $1\le j < i$, the $j$-th  GHW $d_j^{\bot} \ge j(r+1).$ According to Lemma \ref{thm-ghw-upper-bound-dual},  $d_j^{\bot} \le j(r+1).$ Thus the conclusion follows.
\end{proof}

\smallskip

\begin{Lemma} \label{lemma-gap-number}
Let  $\mathcal{C}$ be an $(n,k,r)$ LRC and $\mathcal{C}^{\bot}$ be its dual code. Then the gap numbers of the $\mathcal{C}^{\bot}$ satisfy
 \begin{equation}\label{gap-number-bound}
   g_i^{\bot} \ge \left \lceil \frac{i}{r} \right \rceil + i - 1, \quad \mbox{$1 \le i \le k$}.
\end{equation}
\end{Lemma}

\begin{proof}
The lemma is proved by induction. For $i=1$, it follows that $g_1^{\bot} =1$, since that if $g_1^{\bot} \ge 2$, then $d_1^{\bot} = 1$, which implies $\mathcal{C}$ contains a coordinate which is always zero for all the codewords of $\mathcal{C}$.
Suppose that for $1\le i \le k-1$, $g_i^{\bot} \ge \left \lceil \frac{i}{r} \right \rceil + i - 1$. Then consider the gap number $g_{i+1}^{\bot}$,
\smallskip
\begin{enumerate}
  \item If $\left \lceil \frac{i+1}{r} \right \rceil = \left \lceil \frac{i}{r} \right \rceil$, by $g_{i+1}^{\bot} \ge g_i^{\bot} + 1$, we have $g_{i+1}^{\bot} \ge (\left \lceil \frac{i}{r} \right \rceil + i - 1) + 1  = \left \lceil \frac{i+1}{r} \right \rceil + (i+1)-1$.
  \smallskip
  \smallskip
  \item  If $\left \lceil \frac{i+1}{r} \right \rceil = \left \lceil \frac{i}{r} \right \rceil+1$, i.e., $i=ar, i+1=ar+1$. By Lemma \ref{thm-ghw-upper-bound-dual}, $d_a^{\bot} \le a(r+1)$. Hence the set $\{1,\cdots, a(r+1)\}$ contains at least $\{d_1^{\bot},\cdots, d_a^{\bot}\}$. In other words, the  set $\{1,\cdots, a(r+1)\}$ contains at most these gap numbers $\{g_1^{\bot} \cdots g_{ar}^{\bot}\}$. Then it follows that  $g_{ar+1}^{\bot} \ge a(r+1)+1$, i.e., $g_{i+1}^{\bot} \ge \left \lceil \frac{i+1}{r} \right \rceil + i$.
\end{enumerate}
Combining the above two cases, the conclusion follows.
\end{proof}

\smallskip
\begin{Theorem} \label{thm-weight-hierarchy}
Let  $\mathcal{C}$ be an $(n,k,r)$ LRC, the $i$-th ($1 \le i \le k$) GHW  of $\mathcal{C}$ satisfies the  generalized Singleton-like bound
\begin{equation} \label{weight-hierarchy}
d_i \le n-k-\left \lceil \frac{k-i+1}{r} \right \rceil + i + 1.
\end{equation}
\end{Theorem}

\begin{proof}
For $1 \le i \le k$, the $i$-th GHW satisfies
\begin{eqnarray*}
  d_i &\overset{(a)}{=}& n+1 - g_{k-i+1}^{\bot} \\
      &\overset{(b)}{\le}& n+1 - \left (\left \lceil \frac{k-i+1}{r} \right \rceil + k-i \right ) \\
      &=& n-k-\left \lceil \frac{k-i+1}{r} \right \rceil + i + 1,
\end{eqnarray*}
where ($a$) follows from (\ref{gap-number-dual}) and ($b$) follows from Lemma \ref{lemma-gap-number}.
\end{proof}

\bigskip
\begin{Remark}
Note  that when $i=1$, the generalized Singleton-like  bound (\ref{weight-hierarchy}) gives  the Singleton-like bound (\ref{singleton-like-bound}).
\end{Remark}

\bigskip

\begin{Remark}
 When $r=k$, i.e., there is no locality constraint on the code symbols of $\mathcal{C}$, the generalized Singleton-like bound (\ref{weight-hierarchy}) reduces to the  classical generalized Singleton bound $d_i \le n-k+i$.
\end{Remark}
\section{The GHWs of Optimal LRCs } \label{sec:lower-bound}
In this section, we will focus on the GHWs of optimal LRCs meeting the Singleton-like bound (\ref{singleton-like-bound}). It is shown that for  an optimal $(n,k,r)$  LRC $\mathcal{C}$  with $r \mid k$,  the  weight hierarchy of $\mathcal{C}$ and its dual code $\mathcal{C}^{\bot}$ can be completely determined. For the other  case that $r \nmid k$, we give some lower bounds on the GHWs.

\subsection{The Weight Hierarchy of Optimal LRCs with $r \mid k$}
\begin{Theorem}\label{thm:weight-hierarchy-dual-code}
Let $\mathcal{C}$  be an optimal $(n,k,r)$ LRC meeting the Singleton-like bound (\ref{singleton-like-bound}) with $r \mid k$ and $\mathcal{C}^{\bot}$ be its dual code. The weight hierarchy of $\mathcal{C}^{\bot}$ can be completely determined as
\begin{eqnarray}\label{weight-hierarchy-dual-code}
  d_i^{\bot}=
  \left\{
  \begin{array}{ll}
  i(r+1), \quad & \mbox{$1 \le i \le k/r$},\\
  k+i, \quad &\mbox{$ k/r + 1 \le i \le n-k$}.
  \end{array}
  \right.
\end{eqnarray}
\end{Theorem}
\begin{proof}
 By the duality property (\ref{gap-number-dual}),
\begin{equation}
    d_i^{\bot}  =  (n+1) - g_{n-k-i+1}, \; \; 1 \le i \le n-k.
\end{equation}
Since $d_1 = d = n-k-  k/r + 2$,  $\mathcal{C}$ has gap numbers  $g_s = s$, where $1 \le s \le n-k-  k/r + 1$. Hence by (\ref{gap-number-dual}), it follows that for $k/r \le i \le n-k$, $d_i^{\bot}  =  k + i.$
Since $d_{k/r}^{\bot}  =  k + k/r = \frac{k}{r}(r+1),$ combining Lemma \ref{coro-ghw}, it holds that  for $1 \le i \le k/r$,  $d_i^{\bot}  =  i(r+1).$  Hence the conclusion follows.
\end{proof}

\smallskip
Since the duality properties  (\ref{weight-hierarchy-duality}) and (\ref{gap-number-dual}), we can obtain the following theorem.

\smallskip

\begin{Theorem}\label{thm:weight-hierarchy}
Let $\mathcal{C}$  be an optimal $(n,k,r)$ LRC meeting the Singleton-like bound (\ref{singleton-like-bound})  with $r \mid k$.  The weight hierarchy of $\mathcal{C}$ can be completely determined as
\begin{equation}
  d_i= n-k-\left \lceil \frac{k-i+1}{r} \right \rceil + i + 1,  \; \; 1 \le i \le k.
\end{equation}
\end{Theorem}
\begin{proof}
 By Theorem \ref{thm:weight-hierarchy-dual-code}, it is not hard to verify that the gap numbers of the $\mathcal{C}^{\bot}$ are
 \begin{equation}
   g_i^{\bot} = i + \left \lceil \frac{i}{r} \right \rceil  - 1, \quad \mbox{$1 \le i \le k$}.
 \end{equation}
Then by the duality property (\ref{gap-number-dual}), it follows that for $1 \le i \le k$,
\begin{eqnarray*}
   d_i = n+1 - g_{k-i+1}^{\bot} = n-k-\left \lceil \frac{k-i+1}{r} \right \rceil + i + 1,
\end{eqnarray*}
which completes the proof.
\end{proof}

\bigskip
\begin{Remark}
For an $(n,k,r)$ optimal LRC  meeting the Singleton-like bound (\ref{singleton-like-bound})  with $r \mid k$, the $i$-th ($1 \le i \le k$) GHW  attains the  generalized Singleton-like bound (\ref{weight-hierarchy}) with equality.
\end{Remark}

\bigskip
\begin{Example}
For an $(n=12,k=6,r=3)$ optimal cyclic LRC \cite{Tamo-cyclic} $\mathcal{C}$ with  $r \mid k$, by Theorem  \ref{thm:weight-hierarchy-dual-code}, the weight hierarchy of its dual code $\mathcal{C}^{\bot}$  can be completely determined as
\begin{equation*}
  d_1^{\bot} = 4, d_2^{\bot} = 8, d_3^{\bot} = 9, d_4^{\bot} = 10, d_5^{\bot} = 11, d_6^{\bot} = 12.
\end{equation*}
By Theorem  \ref{thm:weight-hierarchy}, the complete weight hierarchy of $\mathcal{C}$ is
\begin{equation*}
  d_1 = 6, d_2 = 7, d_3 = 8, d_4 = 10, d_5 = 11, d_6 = 12.
\end{equation*}
\end{Example}

\subsection{Lower Bounds on the GHWs of Optimal LRCs}
\begin{Lemma} \label{thm-ghw-dual-lower-bound}
Let $\mathcal{C}$  be an optimal $(n,k,r)$ LRC meeting the Singleton-like bound (\ref{singleton-like-bound}) and $\mathcal{C}^{\bot}$ be its dual code. The $i$-th  GHW of the dual code $\mathcal{C}^{\bot}$ satisfies
\begin{eqnarray}\label{weight-hierarchy-dual-code}
  d_i^{\bot} &\ge& i(r+1) - \left \lceil  k/r  \right \rceil r + k, \   \mbox{$1 \le i \le \lceil  k/r  \rceil-1$}, \\
  d_i^{\bot} &=& k + i,       \hspace{2.69cm}  \mbox{$i \ge \lceil  k/r  \rceil$}.
\end{eqnarray}
\end{Lemma}

\begin{proof}
Since  the $1$st GHW of $\mathcal{C}$ is $d_1 = n-k-\lceil k/r \rceil + 2$, $g_{n-k-\lceil k/r \rceil+1} = n-k-\lceil k/r \rceil+1$.
Then By  (\ref{gap-number-dual}),
\begin{equation}\label{lemma-lb-eq1}
  d_{\lceil k/r \rceil}^{\bot} = (n+1) - g_{n-k-\lceil k/r \rceil+1} = k + \lceil k/r \rceil.
\end{equation}
Hence for $ \lceil k/r \rceil \le i \le n-k$,   $d_i^{\bot} = k + i$. In the meantime, By Lemma \ref{lemma-ghw-lower-bound-dual}, It follows that for $1 \le i \le \lfloor  k/r  \rfloor$,
\begin{equation}\label{lemma-lb-eq2}
  d_{i}^{\bot} \ge d_{i+1}^{\bot} - (r+1).
\end{equation}
Combing (\ref{lemma-lb-eq1}) and (\ref{lemma-lb-eq2}), it follows that for  $1 \le i \le \lceil  k/r  \rceil - 1$,
\begin{eqnarray}
  d_{i}^{\bot}  \ge d_{\lceil k/r \rceil}^{\bot} - (\lceil k/r \rceil-i)(r+1) = i(r+1)  - \left \lceil  k/r  \right \rceil r + k, \nonumber
\end{eqnarray}
which completes the proof.
\end{proof}

\smallskip
\begin{Remark}
Combining Lemma \ref{thm-ghw-upper-bound-dual}, \ref{thm-ghw-dual-lower-bound} with the condition $r \mid k$, we can also obtain the Theorem \ref{thm:weight-hierarchy-dual-code}, which appears to be another proof of the Theorem \ref{thm:weight-hierarchy-dual-code}.
\end{Remark}

\smallskip
\begin{Lemma} \label{lemma-gap-number-upper-bound}
Let $\mathcal{C}$  be an optimal $(n,k,r)$ LRC meeting the Singleton-like bound (\ref{singleton-like-bound}) and $\mathcal{C}^{\bot}$ be its dual code. Then the gap numbers of the $\mathcal{C}^{\bot}$ satisfy
 \begin{equation}\label{gap-number-bound}
   g_i^{\bot} \le \left \lceil \frac{i+(\lceil \frac{k}{r} \rceil r-k)}{r} \right \rceil + i - 1, \quad \mbox{$1 \le i \le k$}.
\end{equation}
\end{Lemma}
\begin{proof}
The lemma is proved by induction. For $i=k$, it follows that $g_k^{\bot} =(n+1)-d_1= k+\lceil k/r \rceil-1$.
Suppose that for $2 \le i \le k$, $g_i^{\bot} \le \left \lceil \frac{i+(\lceil \frac{k}{r} \rceil r-k)}{r} \right \rceil + i - 1$. Then consider the gap number $g_{i-1}^{\bot}$,
\smallskip
\begin{enumerate}
  \item If $\left \lceil \frac{i-1+(\lceil \frac{k}{r} \rceil r-k)}{r} \right \rceil = \left \lceil \frac{i+(\lceil \frac{k}{r} \rceil r-k)}{r} \right \rceil $, we have $g_{i-1}^{\bot} \le g_{i}^{\bot}-1 = \left \lceil \frac{i-1+(\lceil \frac{k}{r} \rceil r-k)}{r} \right \rceil + (i-1) - 1$.
  \smallskip
  \smallskip
  \item  If $\left \lceil \frac{i-1+(\lceil \frac{k}{r} \rceil r-k)}{r} \right \rceil = \left \lceil \frac{i+(\lceil \frac{k}{r} \rceil r-k)}{r} \right \rceil-1$, i.e., $(i-1)+(\lceil \frac{k}{r} \rceil r-k)=ar, i+(\lceil \frac{k}{r} \rceil r-k)=ar+1$. By Lemma \ref{thm-ghw-dual-lower-bound}, $d_{\lceil \frac{k}{r} \rceil}^{\bot} = k+\lceil \frac{k}{r} \rceil$ and for $1 \le i \le \lceil \frac{k}{r} \rceil-1$, $d_i^{\bot} \ge i(r+1) - \left \lceil \frac{k}{r} \right \rceil r + k.$  Hence the set $\{1,\cdots, a(r+1)-(\left \lceil \frac{k}{r} \right \rceil r - k)-1\}$ contains at most $\{d_1^{\bot},\cdots, d_{a-1}^{\bot}\}$. In other words, the  set $\{1,\cdots, a(r+1)-(\left \lceil \frac{k}{r} \right \rceil r - k)-1\}$ contains at least these gap numbers $\{g_1^{\bot} \cdots g_{ar-(\left \lceil \frac{k}{r} \right \rceil r - k)}^{\bot}\}$. Then it follows that  $g_{ar-(\left \lceil \frac{k}{r} \right \rceil r - k)}^{\bot} \le a(r+1)-(\left \lceil \frac{k}{r} \right \rceil r - k)-1$, i.e., $g_{i-1}^{\bot} \le \left \lceil \frac{i-1+(\lceil \frac{k}{r} \rceil r-k)}{r} \right \rceil  + (i-1)-1$.
\end{enumerate}
Combining the above two cases, the conclusion follows.
\end{proof}
\smallskip
\begin{Theorem} \label{thm-ghw-lower-bound}
Let  $\mathcal{C}$ be an optimal $(n,k,r)$ LRC meeting the Singleton-like bound (\ref{singleton-like-bound}), its $i$-th ($1 \le i \le k$) GHW
\begin{equation}
d_i \ge n-k-\left \lceil \frac{\lceil \frac{k}{r} \rceil r-i+1}{r} \right \rceil + i + 1.
\end{equation}
\end{Theorem}

\begin{proof}
For $1 \le i \le k$, the $i$-th GHW satisfies
\begin{eqnarray*}
  d_i &\overset{(a)}{=}& n+1 - g_{k-i+1}^{\bot} \\
      &\overset{(b)}{\ge}& n+1 - \left (\left \lceil \frac{k-i+1+\lceil \frac{k}{r} \rceil r-k}{r} \right \rceil + k-i \right) \\
      &=& n-k-\left \lceil \frac{\lceil \frac{k}{r} \rceil r-i+1}{r} \right \rceil + i + 1,
\end{eqnarray*}
where ($a$) follows from  (\ref{gap-number-dual}) and ($b$) follows from Lemma \ref{lemma-gap-number-upper-bound}.
\end{proof}

\smallskip
\begin{Remark}
Combining Theorem \ref{thm-weight-hierarchy}, \ref{thm-ghw-lower-bound} with the condition $r \mid k$, Theorem \ref{thm:weight-hierarchy} can also be obtained.
\end{Remark}

\section{Bounds of Linear Codes in Terms of the GHWs} \label{sec:general-bound}
In this section, we employ a parity-check matrix approach \cite{Hao} to present two general bounds of  linear codes in terms of the GHWs. The general bounds take the field size into account and can be used to derive several known bounds of  LRCs. It is also shown that the result on the bounds of LRCs in \cite{Tamo-matroid} can  be explained in terms of the GHWs. Some result in \cite{Wang:IT} can be further refined by using the GHWs.

\smallskip
\begin{Proposition} \label{prop:general-d-ghw}
Let $\mathcal{C}$ be an $[n,k,d]$ linear code and $\mathcal{C}^{\bot}$ be its dual code. The minimum distance of $\mathcal{C}$  satisfies
\begin{equation} \label{label:general-bound-ghw}
d \le \min_{ 1 \leq i \leq g_k^{\bot} -k} \; d^{(q)}_{\rm opt}(n  - d_i^{\bot},k+i- d_i^{\bot}),
\end{equation}
where $d_{\text{opt}}^{(q)}(n^*,k^*)$ is the largest possible minimum distance of a $q$-ary linear code with length $n^*$ and dimension  $k^*$, and $d_i^{\bot}$ is the $i$-h GHW of $\mathcal{C}^{\bot}$.
\end{Proposition}
\begin{proof}
By (\ref{g_k_1}), we obtain that
\begin{equation}\label{g_k}
  g_k^{\bot} -k = {\rm{max}}\{i: 1\le i \le n-k, d_i^{\bot} < k+i\}.
\end{equation}
In other words, for any $1\le i \le g_k^{\bot} -k$, it follows that $ k+i-d_i^{\bot}  > 0$. Now consider a subcode $\mathcal{D}_i^{\bot}$ of $\mathcal{C}^{\bot}$  with the support $|{\rm supp}(\mathcal{D}_{i}^{\bot})|=d_i^{\bot}$. Let $H_1=[(\mathbf{h}_1^T, \cdots \mathbf{h}_i^T)^T]$ be a basis of  the subcode $\mathcal{D}_i^{\bot}$ with $i$ independent parity-check equations from $\mathcal{D}^{\bot}_i$.

Now let us select $n-k$ independent parity-check equations from $\mathcal{C}^{\bot}$ to construct a parity-check matrix $H=(H_1^T, H_2^T)^T$ of $\mathcal{C}$, where the first $i$ parity-check equations are the vectors from $H_1$. The lower part $H_2$ consists of other $n-k-i$ independent parity-check equations form $\mathcal{C}^{\bot}$.
The first $i$ rows in $H_1$ cover $|{\rm supp}(\mathcal{D}_{i}^{\bot})|$ columns of $H$. By deleting the first $i$ rows and the corresponding $|{\rm supp}(\mathcal{D}_{i}^{\bot})|$ columns of $H$, we have an $m^*\times n^*$ sub-matrix $H^*$, where $m^*= n-k-i$ and $n^* = n - |{\rm supp}(\mathcal{D}_{i}^{\bot})| = n - d_i^{\bot}$. Let $\mathcal{C}^*$ be the $[n^*,k^*,d^*]$ linear code with parity-check matrix $H^*$.
Among the $n^*$ columns of $H$, since the elements lies above $H^*$ are all zero, $d \leq d^*$.
Moreover, by ${\rm rank}(H^*)\leq  n-k-i$, we have  $k^* = n^*-{\rm rank}(H^*) \geq   k+i - d_i^{\bot}>0$. Hence,
\begin{eqnarray}
 d\;\le \;d^*  \;\leq\; d^{(q)}_{\rm opt}(n^*,k^*) \;\leq\; d^{(q)}_{\rm opt}(n - d_i^{\bot},k+i - d_i^{\bot}). \label{C*}
\end{eqnarray}
Since  $1 \leq i \leq  g_k^{\bot} -k$, the conclusion follows.
\end{proof}

\smallskip
In the proof of Proposition $\ref{prop:general-d-ghw}$, the next result follows by substituting (\ref{C*}) by
\begin{eqnarray}
k\;\le \;k^*-i+d_i^{\bot}  \;&\leq&\; k^{(q)}_{\rm opt}(n^*,d^*)-i+d_i^{\bot} \nonumber \\
   &\leq& \; k^{(q)}_{\rm opt}(n -d_i^{\bot},d)-i+d_i^{\bot}.
\end{eqnarray}

\begin{Proposition}\label{prop:general-k-ghw}
Let $\mathcal{C}$ be an $[n,k,d]$ linear code and $\mathcal{C}^{\bot}$ be its dual code. The dimension of $\mathcal{C}$ satisfies
\begin{equation*}\label{label:general-bound-k}
k \le \min_{ 1 \leq j \leq g_k^{\bot} -k} \left[ k^{(q)}_{\rm opt}(n -d_i^{\bot},d)-i+d_i^{\bot} \right],
\end{equation*}
where $k_{\text{opt}}^{(q)}(n^*,d^*)$ is the largest possible dimension of a $q$-ary linear code with length $n^*$ and minimum distance  $d^*$ and $d_i^{\bot}$ is the $i$-h GHW of $\mathcal{C}^{\bot}$.
\end{Proposition}

\smallskip
\begin{Remark}
As for $(n,k,r)$ LRCs,  By Lemma \ref{lemma-gap-number}, $g_k^{\bot} \ge \left \lceil k/r \right \rceil + k - 1$. Hence $g_k^{\bot} -k \ge \left \lceil k/r  \right \rceil  - 1$. By Lemma \ref{thm-ghw-upper-bound-dual}, for $1 \le i \le \lfloor  k/r  \rfloor$,  $d_i^{\bot} \le i(r+1)$. In the above proof, for $1 \le i \le \left \lceil  k/r  \right \rceil  - 1$, if we delete $i(r+1)-|{\rm supp}(\mathcal{D}_{i}^{\bot})|$ columns more, then $n^* = n - i(r+1)$. In this case Proposition \ref{prop:general-d-ghw} gives the following bound \cite{Hao}
\begin{equation} \label{label:general-bound-ghw}
d \le \min_{ 1 \leq i \leq \lceil \frac{k}{r} \rceil -1 } \; d^{(q)}_{\rm opt}(n  - i(r+1),k-ir).
\end{equation}
where $d_{\text{opt}}^{(q)}(n^*,k^*)$ is the largest possible minimum distance of a $q$-ary linear code with length $n^*$ and dimension  $k^*$.
Similarly, Proposition \ref{prop:general-k-ghw} gives the Cadambe-Mazumdar bound \cite{Cadambe:IT}
\begin{equation}\label{cm-bound}
k \le \min_{  1 \leq i \leq \lceil \frac{k}{r} \rceil -1} \Big[ ir + k^{(q)}_{\rm opt}(n  - i(r+1),d) \Big],
\end{equation}
where $k_{\text{opt}}^{(q)}(n^*,d^*)$ is the largest possible dimension of an $n^*$-length code  given the alphabet size $q$ and distance $d^*$,
\end{Remark}

\bigskip

For a linear code $\mathcal{C}$, by (\ref{g_k_1})(\ref{g_k_2}), the $k$-th gap number of  $\mathcal{C}^{\bot}$
\begin{eqnarray}
  g_k^{\bot} = k+ {\rm{max}}\{i: 1\le i \le n-k, d_i^{\bot} < k+i\},  \hspace{0.8cm} \label{g_k_3}  \\
    = k+ {\rm{min}}\{i: 1\le i \le n-k, d_i^{\bot} = k+i\} -1. \hspace{0.33cm}  \label{g_k_4}
\end{eqnarray}
Hence
\begin{eqnarray}
  d  = n+1 - g_k^{\bot} \hspace{5.65cm}  \nonumber\\
   \hspace{0.5cm}  =n-k - {\rm{max}}\{i: 1\le i \le n-k, d_i^{\bot} < k+i\} +1, \label{g_k_5}  \\
    =n- k - {\rm{min}}\{i: 1\le i \le n-k, d_i^{\bot} = k+i\}+2.\hspace{0.12cm} \label{g_k_6}
\end{eqnarray}
In the following, we will show that the results on the minimum distances of linear codes in \cite{Tamo-matroid}\cite{Wang:IT} can be  explained or refined by using the properties (\ref{g_k_5})(\ref{g_k_6}) of the gap number  $g_k^{\bot}$.

\smallskip
Tamo \emph{et al.} introduced a matroid approach to study the optimality of the minimum distance for a linear code in \cite{Tamo-matroid}.  Let $G$ be the $k \times n$ generator matrix of a linear code $\mathcal{C}$. Define the matroid $\mathcal{M}([n],{\rm rank}(.))$, the rank of a set $\mathcal{A} \ \subseteq [n]$ is ${\rm rank}(\mathcal{A}) = {\rm rank} (G_{\mathcal{A}})$, where $G_{\mathcal{A}}$ is the submatrix of $G$ indexed by $\mathcal{A}$ and rank is the normal operator on linear vectors. A  set $\mathcal{B} \ \subseteq [n]$ is said to be a \emph{circuit} if  ${\rm rank}(\mathcal{B}) =|\mathcal{B}|-1$ and all its proper subsets are independent. Tamo \emph{et al.} introduced the notion of \emph{nontrivial union} of circuits and proposed the following result on the minimum distance of linear codes.
\begin{Proposition} \label{prop:tamo}
\cite[Theorem 2]{Tamo-matroid} Let $G$, $\mathcal{M}$ and $\mu$ be defined as above, then $\mathcal{C}$  has the minimum distance
\begin{equation}\label{tamo-thm}
  d = n-k-\mu + 2,
\end{equation}
where $\mu$ is the minimum positive integer such that the size of every nontrivial union of $\mu$ circuits in $\mathcal{M}$ is at least $k + \mu$.
\end{Proposition}

\smallskip
According to the definition of the circuits, the columns of $G$ indexed by the coordinates in a circuit must be linearly dependent. Hence it is not hard to see that a circuit corresponds to a parity-check equation in $\mathcal{C}^{\bot}$.  A nontrivial union \cite{Tamo-matroid} of $\nu$ circuits corresponds to $\nu$ linearly independent parity-check equations, which corresponds to a subcode $\mathcal{D}^{\bot}$ in $\mathcal{C}^{\bot}$ with dimension $\nu$.
By the definition, the size of every nontrivial union of $\nu$ circuits in $\mathcal{M}$ is at least $k + \nu$, hence it follows that $d_{\nu}^{\bot} \ge k+\nu.$
On the other hand, by the generalized Singleton bound, $d_{\nu}^{\bot} \le k+\nu.$ Thus  $d_{\nu}^{\bot} = k+\nu.$ Then combining  the definition that $\mu$ is the minimal value of $\nu$, we have
 \begin{equation}
  \mu = {\rm{min}}\{\nu: 1\le \nu \le n-k, d_{\nu}^{\bot} = k+\nu\}.
\end{equation}
Hence by (\ref{g_k_6}), Proposition \ref{prop:tamo} follows.

\bigskip

Wang \emph{et al.} proposed a framework of regenerating sets \cite{Wang:ISIT,Wang:IT} which extends the matroid approach in \cite{Tamo-matroid} to include the vector case and nonlinear  case. For the linear case of this framework, Wang \emph{et al.} obtain the following result.
\begin{Proposition} \label{prop:wang}
\cite[Theorem 2]{Wang:IT} For an $[n,k,d]$ linear code,
\begin{equation}\label{wang-thm}
  d \le n-k-\rho + 1,
\end{equation}
where $\rho =  {\rm{max}}\{x: \Phi(x)-x < k\},$  $\Phi(x)={\rm{min}}\{|\cup_{j=1}^x R_j|:R_1, \cdots, R_x  \; \mbox{\rm has a nontirvial union}\}$.
\end{Proposition}

\smallskip
Similarly, according to the definition of the regenerating set and the nontrivial union \cite{Wang:IT}, we can see that a regenerating set corresponds to a parity-check equation in $\mathcal{C}^{\bot}$ and a nontrivial union of $x$ regenerating sets corresponds to a subcode $\mathcal{D}^{\bot}$ in $\mathcal{C}^{\bot}$ with dimension $x$.  Hence in terms of GHWs, $\Phi(x) = d_x^{\bot}$.  Then by (\ref{g_k_1}),
\begin{equation}
 \rho= {\rm{max}}\{x: d_x^{\bot}-x < k, 1\le x \le n-k\} = g_k^{\bot} - k.
\end{equation}
Thus by (\ref{g_k_5}) we can obtain that
\begin{equation}
  d =n- k  - \rho +1,
\end{equation}
 which gives a refined explanation of the Proposition \ref{prop:wang} in terms of the GHWs where (\ref{wang-thm}) holds with equality.

\section{Conclusions}
\label{sec:conclusion}
In this paper, we  studied the GHWs of LRCs. A generalized Singleton-like bound on the GHWs of  LRCs were given. For an optimal $(n,k,r)$ LRC $\mathcal{C}$ meeting the Singleton-like bound (\ref{singleton-like-bound}) with $r \mid k$, the weight hierarchy of $\mathcal{C}$ and its dual code $\mathcal{C}^{\bot}$  was completely determined. Moreover, the $i$-th GHW of an optimal $(n,k,r)$  LRC with  $r \mid k$ attains the generalized Singleton-like bound for all $1\le i \le k$. For an optimal $(n,k,r)$ LRC with $r \nmid k$, we gave lower bounds on the GHWs of the optimal LRC and its dual code. At last, two general bounds on linear codes in terms of the GHWs were presented. It was also shown that the result on the bounds of minimum distances of linear codes \cite{Tamo-matroid} can also be explained by using the $k$-th gap number of the dual codes, and the result in \cite{Wang:IT} can be further refined by using the $k$-th gap number of the dual codes.



%

\end{document}